\documentclass[10pt,twocolumn]{article} 
\usepackage{simpleConference}
\usepackage{times}
\usepackage{graphicx}
\usepackage{amssymb}
\usepackage{url,hyperref}
\usepackage{amsmath}
\usepackage{amsfonts}
\usepackage{amsthm}
\usepackage{mathtools}
\usepackage{physics}
\usepackage{bbm}
\usepackage{amssymb}

\newcommand{\cosimcid}{\mathcal{C}_{\text{id}}^{\text{sim}}}
\newcommand{\cocid}{\mathcal{C}_{\text{id}}}
\newcommand{\simcid}{C_{\text{id}}^{\text{sim}}}
\newcommand{\cid}{C_{\text{id}}}

\newcommand{\alx}{\mathcal{X}}
\newcommand{\aly}{\mathcal{Y}}
\newcommand{\hi}{\mathcal{H}}
\newcommand{\st}{\mathcal{S}}

\newcommand{\co}{\mathcal{C}}

\newcommand{\ciel}[1]{\lceil #1 \rceil}
\newcommand{\ou}[1]{\overline{\mathcal{#1}}}
\newcommand{\nd}{\text{ and }}

\newcounter{thm}
\numberwithin{thm}{section}
\newtheorem{theorem}[thm]{Theorem}
\theoremstyle{theorem}

\newtheorem{remark}[thm]{Remark}
\newtheorem{definition}[thm]{Definition}
\newtheorem{corollary}[thm]{Corollary}
\newtheorem{lemma}[thm]{Lemma}

\newtheorem{notation}[thm]{Notation}

\allowdisplaybreaks

\begin{document}
\title{The Simultaneous Identification Capacity of the Classical - Quantum Multiple Access Channel} 
	
\author{Holger Boche, Stephen Diadamo\\
	\text{\small{[boche, stephen.diadamo]@tum.de}}
	\\\\
	Lehrstuhl f\"ur Theoretische Informationstechnik, Technische Universit\"at M\"unchen\\
	80290 Munchen, Germany 
}
\maketitle
\thispagestyle{empty}
	
\begin{abstract}
	Here we discuss message identification, a problem formalized by Rudolf Ahlswede and Gunter Dueck, over a classical-quantum multiple access channel with two classical senders and one quantum receiver. We show  that the simultaneous identification capacity, a capacity defined by Peter L\"ober, of this  multiple access channel is equal to its message transmission capacity region. 
\end{abstract}
	
\section{Introduction}
	
In message transmission, the receiver of the message is interested to know what exactly the message received reads. This differs from message identification such that the receiver of the message is only looking to answer the question, ``Is this received message the one I am interested in?" The sender is free to choose the message they wish to send, and it is up to the receiver to determine this single bit of information.
	
Over classical channels it was shown by Rudolf Ahlswede, Gunter Dueck, et al that there exists identification codes that are doubly exponential in size with respect to the blocklength. Further, for classical-quantum channels, the result was extended by Peter L\"ober to show that there exists simultaneous identification codes which also are doubly exponential in size with respect to blocklength. L\"ober includes a restriction on his codes such that they must be simultaneous, which we provide the definition for in a latter section, but it was shown by Andreas Winter and Rudolf Ahlswede that this restriction can be dropped and so the single logarithmic scaled transmission capacity of a discrete, memoryless classical-quantum channel is equal to its doubly logarithmic scaled identification capacity.

Here we will firstly consider the classical-quantum multiple access channel with two classical senders and one quantum receiver (CCQ). We begin by introducing the framework for working with classical-quantum channels and review definitions for the message transmission setting.	
	
Next we provide theory for the transmission capacity region of the CCQ channel using codes under a maximal error error criterion. In the case of the single sender-single receiver channel, there is no difference in capacity between using codes under a maximal error error criterion or an average error criterion, but it is known in the classical case that these two capacity regions are not equal \cite{capacity_regions_different}. Because use of maximal error codes are made in the proof of achievability, we need that the capacity region is not empty when the average error capacity region is non-empty, and this is therefore proved. 
	
In the following section, we define the models for identification over a CCQ channel. We provide definitions for a restricted version of the problem, namely, simultaneous identification, a concept by Peter L\"ober \cite{loeber_id_cap}. Because the channel outputs a quantum state, multiple measurements for identification on the state cannot be performed. The simultaneous identification model overcomes this by performing all identification measurements at one time with a single measurement. With this, we can define a simultaneous identification capacity region for the CCQ channel and further determine a quantity for this region.
	
The achievability proof follows the technique of Ahlswede and Dueck to prove the achievability of the classical single sender-single receiver channel, that is, adding a small amount of randomness to a transmission code to transform it into a random identification (ID) code. We show that there exists realizations of the random ID code that achieve the desired capacity.
	
For the converse, we follow the technique of Peter L\"ober  \cite{loeber_id_cap} and Yosef Steinberg \cite{steinberg} which uses resolvability theory, formalized by Te Sun Han and Sergio Verd\'u \cite{resolve_theory}. With this converse, we conclude that the simultaneous identification capacity of the CCQ multiple access channel is indeed equal to the transmission capacity region for memoryless CCQ channels.
	
\section{Review of Message Transmission}
	
In this section, we review the definitions for classical-quantum channels and their codes. By $\alx$ or $\aly$ we refer to finite alphabets. The $k$-fold product set of an alphabet $\alx$ is denoted as $\alx^k \coloneqq \alx \times ... \times \alx$. The space of quantum states with respect to a particular Hilbert space $\hi$ is denoted as $\mathcal{S}(\hi)$ and the set of linear operators on $\hi$ is denoted as $\mathcal{L}(\hi)$. Here we only consider finite dimensional, complex Hilbert spaces. The set of probability distributions on another set $\alx$ is denoted as $\mathcal{P}(\alx)$. Quantum channels, usually denoted as $W$, in this report are always completely positive and trace preserving maps. 
	
\begin{definition}[CQ channel]
	A classical-quantum (CQ) channel $\mathbf{W}$ is a family of quantum channels 
	\begin{align*}
		\{ W^k : \alx^k \rightarrow \mathcal{S}(\mathcal{H}^{\otimes k})\}_{k\in\mathbb{N}}. 
	\end{align*}
	We say that $\mathbf{W}$ is a discrete memoryless CQ (DM-CQ) channel generated by $W$ if for all $x^k = (x_1,...,x_k) \in \alx^k$, for all $k$, 
	\begin{align*}
		W^{k}(x^k) = \bigotimes_{i=1}^k W(x_i). 
	\end{align*}
\end{definition}

\begin{definition}[CCQ channel]
	A CCQ channel $\mathbf{W}$ is a family of quantum channels  
	\begin{align*}
		\{ W^k : \alx^k \times \aly^k \rightarrow \mathcal{S}(\mathcal{H}^{\otimes k})\}_{k\in\mathbb{N}}. 
	\end{align*}
	We say that $\mathbf{W}$ is a DM-CCQ channel generated by $W$ if for all $x^k = (x_1,...,x_k) \in \alx^k$ and all $y^k = (y_1,...,y_k)\in\aly^k$, for all $k$, 
	\begin{align*}
		W^{k}(x^k, y^k) = \bigotimes_{i=1}^k W(x_i,y_i). 
	\end{align*}
\end{definition}

\begin{notation}
	If $\mathbf{W}$ is a discrete memoryless channel generated by a channel $W$, we refer to it simply as $W$, that is, without the bold face. We make this distinction since some lemmas or theorems will hold in general, but some are proved with the assumption of a discrete memoryless property. 
\end{notation}

\begin{definition}[Channel state]\label{gamma}
	For CCQ channel $\mathbf{W}$, probability distributions $p_1 \in \mathcal{P}(\alx^k)$ and $p_2 \in \mathcal{P}(\aly^k)$, and Hilbert spaces $\hi_A$ and $\hi_B$ with respective orthonormal bases $\left \{\ket{x^k}\right \}_{x^k\in\alx^k}$ and $\left \{\ket{y^k}\right \}_{y^k\in\aly^k}$, the channel state is defined as 
	\begin{align*}
		\begin{aligned}
		  & \gamma^k_2(p_1,p_2) \coloneqq \\ &\sum_{\substack{x^k \in \alx^k \\ y^k \in \aly^n}} p_1(x^k) p_2(y^k) \ketbra{x^k} \otimes \ketbra{y^k} \otimes W^k(x^k, y^k).
		\end{aligned}
	\end{align*}
	When $k=1$, that is, $p_1\in\mathcal{P}(\alx)$ and $p_2\in\mathcal{P}(\aly)$, we write for notational simplicity $\gamma_2(p_1,p_2) \coloneqq \gamma^1_2(p_1,p_2) $. We also define CQ channel states defined similarly, where we denote $\gamma_1$ and $\gamma_1^k$ for single sender channel states in a similar way, with a single distribution parameter over one orthonormal basis. 
\end{definition}

\begin{definition}[$(k,M,N)$-code]
	For a CCQ channel $\mathbf{W}$, a $(k,M,N)$-code for classical message transmission is the family $\co\coloneqq (x_{m}, y_{n}, D_{mn})_{m=1,n=1}^{M,N}$ where $x_{1},...,x_{M} \in \alx^k$, $y_{1},...,y_{N} \in \aly^k$, and $\{D_{mn}\}_{m=1,n=1}^{M,N} \subset \mathcal{L}(\hi^{\otimes k})$ forms a POVM. \\\\
	For a $(k,M,N)$-code $\co$, the average error of transmission is defined as 
	\begin{align*}
		\overline{e}(\co, W^k) \coloneqq 1 - \frac{1}{MN}\sum_{\substack{m=1 \\n=1}}^{M,N} \tr(D_{mn} W^k(x_{m},y_{n})),
	\end{align*}	
	and the maximal,
	\begin{align*}
		e(\co, W^k) \coloneqq \max_{m,n}\hspace{2mm} 1 - \tr(D_{mn} W^k(x_{m},y_{n})). 
	\end{align*}
\end{definition}

\begin{definition}[Achievable rate pair]
	For a CCQ channel $\mathbf{W}$, we say $(R_1, R_2) \in \mathbb{R}^2$, $R_1, R_2 \geq 0$, is an achievable rate pair if for all $\epsilon, \delta > 0$, there exists a $k_0$ such that for all $k\geq k_0$ there exists a $(k,M,N)$-code $\co$ such that, 
	\begin{align}\label{conds_ach_pair}
		\frac{1}{k}\log M \geq R_1 - \delta, \hspace{4mm} \frac{1}{k} \log N \geq R_2 - \delta, \hspace{4mm} \overline{e}(\co, W^k) \leq \epsilon. 
	\end{align}
	The capacity region for $\mathbf{W}$ is defined as
	\begin{align*}
		C(\mathbf{W}) \coloneqq \{(R_1, R_2) \mid (R_1, R_2) \text{ is an achievable rate pair} \}. 
	\end{align*}
	We say a $(k,M,N)$-code $\co$ achieves the rate pair $(R_1, R_2)$ if (\ref{conds_ach_pair}) holds for all $\epsilon, \delta > 0$.
\end{definition}

\begin{lemma}\label{single_sender}
	For a CCQ channel $\mathbf{W}$, a $(k, M, N)$-code $\co$ with $\overline{e}(\co,W^k) \leq \epsilon$ can be used as a $(k,M)$ or $(k,N)$-code $\co_M$ or $\co_N$ with $\overline{e}(\co_M, W_M^k) =  \overline{e}(\co_N, W_N^k) \leq \epsilon$, where $W_M^k$ and $W_N^k$ are the CQ channels generated by taking the average output over one sender.
\end{lemma}

\begin{proof}
	Let $\co\coloneqq(x_m,y_n,D_{mn})_{m=1,n=1}^{M,N}$ be a $(k, M, N)$-code with  $\overline{e}(\co,W^k) \leq \epsilon$. Consider the channel
	\begin{align*}
		W_M^k & : \alx^k \rightarrow \st(\hi^{\otimes k})            \\
		      & : x^k \mapsto  \frac{1}{N}\sum_{n=1}^N W^k(x^k, y_n) 
	\end{align*} 
	Define $s_m(D_{mn}) \coloneqq \sum_{n=1}^{N} D_{mn}$ and the code $\co_M \coloneqq (x_m,  s_m(D_{mn}))_{m=1}^M$.
	\newpage
	\begin{align*}
		\overline{e} & (W_M^k,\co_M) \\ &= 1 - \frac{1}{M} \sum_{m=1}^M \tr(s_m(D_{mn})W_M^k(x_m)) \\
		&= 1 - \frac{1}{MN} \sum_{n=1}^N\sum_{m=1}^M \tr(\left (\sum_{n'=1}^{N} D_{mn'}\right ) W^k(x_m, y_n)) \\
		&\leq 1 - \frac{1}{MN} \sum_{n=1}^N\sum_{m=1}^M \tr(D_{mn} W^k(x_m, y_n))  \\
		&\leq \epsilon 
	\end{align*}
	
	Thus $\co_M$ is a $(k,M)$-code that has average error bounded by $\epsilon$ over the channel $W_M^k$. Analogous arguments can be made to construct a $(k,N)$-code with bounded average error over a channel 
	\begin{align*}
		W_N^k & : \aly^k \rightarrow \st(\hi^{\otimes k})               \\
		      & : y^k \mapsto  \frac{1}{M}\sum_{n=1}^N W^k(x_m^k, y^k). 
	\end{align*} 
\end{proof}

\section{Maximal Message Transmission Error Capacity Region}
In this section we give a brief analysis of the capacity region defined using the maximal message transmission error figure of merit, rather than the average message transmission error. We refer to this type of merit as the ``maximal-error".

\begin{definition}[Maximal-error achievable rate pair]
	For a CCQ channel $\mathbf{W}$, we say $(R_1, R_2) \in \mathbb{R}^2$, $R_1, R_2 \geq 0$, is a max-error achievable rate pair if for all $\epsilon, \delta > 0$, there exists a $k_0$ such that for all $k\geq k_0$ there exists a $(k,M,N)$-code $\co$ such that,
	\begin{align}\label{conds_max}
		\frac{1}{k}\log M \geq R_1 - \delta, \hspace{4mm} \frac{1}{k} \log N \geq R_2 - \delta, \hspace{4mm} e(\co, W^k) \leq \epsilon. 
	\end{align}
	The max-error capacity region for $\mathbf{W}$ is defined as
	\begin{align*}
		\begin{aligned}
		  & C_{\max}(\mathbf{W})\coloneqq \\ &\hspace{2mm}\{(R_1, R_2) \mid (R_1, R_2) \text{ is a max-error achievable rate pair} \}. 
		\end{aligned}
	\end{align*}
	We say a $(k,M,N)$-code $\co$ ``max-error achieves" the rate pair $(R_1, R_2)$ if (\ref{conds_max}) holds for all $\epsilon, \delta > 0$.
\end{definition}

\begin{lemma}[Maximal error capacity region is closed]
	For CCQ channel $\mathbf{W}$, $C_{\text{max}}(\mathbf{W})$ is closed.
\end{lemma}

\begin{proof}
	Let $(R_{1i}, R_{2i})_{i\in \mathbb{N}}$ be a sequence in $C_{\text{max}}(\mathbf{W})$ that converge to some $(R_1, R_2)$. By the convergence of $(R_{1i}, R_{2i})_{i\in \mathbb{N}}$, there exists an $i_0$ large enough such that for $\delta > 0$,
	\begin{align*}
		R_{1{i_0}} & \geq R_1 - \delta/2 \\  R_{2{i_0}} &\geq R_2 - \delta/2.
	\end{align*}
	Since $(R_{1{i_0}}, R_{2{i_0}})$ is a max-error achievable rate pair, there exists a $k_0$ such that for all $k\geq k_0$, there is a $(k,M,N)$-code $\co$ such that for all $\epsilon >0$, $e(\co,W^k)\leq \epsilon$ and 
	\begin{align*}
		\frac{1}{k}\log M & \geq R_{1{i_0}} - \delta/2 \geq R_1 - \delta \\ \frac{1}{k}\log N &\geq R_{2{i_0}} - \delta/2 \geq R_2 - \delta.
	\end{align*}
	Since $\delta$ was chosen arbitrarily, $\co$ also achieves $(R_1, R_2)$ under maximum error criterion, hence $(R_1, R_2)\in C_{\text{max}}(\mathbf{W})$ and therefore $C_{\max}(\mathbf{W})$ is closed.
\end{proof}

\begin{lemma}[Maximal error capacity region is convex]
	For CCQ channel $\mathbf{W}$, $C_{\text{max}}(\mathbf{W})$ is convex.
\end{lemma}

\begin{proof}
	Let $(R_{1}, R_{2})$ and $(R_{1}', R_{2}')$ be two pairs in $C_{\max}(\mathbf{W})$ and $\epsilon, \delta >0$, then there exists a $k_0$ such that for all $k \geq k_0$,  $(k, M, N)$-code $\co \coloneqq (x_m, y_n, D_{mn})_{m=1, n=1}^{M,N}$ satisfies $e(\co, W^{k}) \leq  \epsilon^{1/2}$, $M\geq 2^{k(R_1 - \delta_1)}$, $N\geq 2^{k(R_2 - \delta_1)}$, and a $k_0'$ such that for all $k'\geq k_0'$, $(k', M', N')$-code  $\co' \coloneqq (x'_m, y'_n, D'_{mn})_{m=1, n=1}^{M',N'}$ satisfies $e(\co', W^{k'}) \leq \epsilon^{1/2}$, $M'\geq 2^{k'(R'_1 - \delta_2)}$, $N'\geq 2^{k'(R'_2 - \delta_2)}$. We make the choice $\delta_1$ and $\delta_2$ such that the rate calculations below are satisfied. We show that for all $\alpha \in [0,1]$ that $(\alpha R_1 + (1-\alpha)R_1', \alpha R_2 + (1-\alpha)R_2')$ is a max-error achievable rate pair. The strategy is to construct a new code that sends first a message $(x_m, y_n)$ from code $\co$ with $k$ uses of the channel and then a message $(x_{m'}', y_{n'}')$ from $\co'$ with $k'$ uses of the channel. The maximum error of such a code is bounded by, 
	\begin{align*}
		\begin{aligned} & \tr( D_{mn} \otimes D'_{m'n'} W^k(x_m, y_n) \otimes W^{k'}(x_{m'}, y_{n'})) \\ 
		                & = \tr( D_{mn} W^k(x_m, y_n))\tr(D'_{m'n'} W^{k'}(x_{m'}, y_{n'}))           \\ 
		                & \leq \epsilon.                                                              
		\end{aligned}
	\end{align*}
	Let $\alpha \in [0,1]$, then 
	\begin{align*}
		\frac{1}{k+k'} & \log(M\cdot M') \\ &\geq \frac{1}{k+k'}\log(2^{k(R_1 - \delta_1)}2^{k'(R'_1 - \delta_2)}) \\
		&= \frac{k}{k+k'}(R_1 - \delta_1) + \frac{k'}{k+k'}(R'_1 - \delta_2) \\
		&= \alpha R_1 + (1-\alpha)R'_1 - \alpha(\delta_1 +\delta_2) + \epsilon_e\\
		&= \alpha R_1 + (1-\alpha)R'_1 - \delta,
	\end{align*} 
	where $k$ and $k'$ can be chosen such that $\frac{k}{k+k'}$ is arbitrarily close to $\alpha$ with $\epsilon_e$ the inaccuracy. We can also choose $\delta_1$ and $\delta_2$  such that $ \delta = \alpha(\delta_1 +\delta_2) - \epsilon_e$. We can make the same arguments to show that $\frac{1}{k+k'}\log(N\cdot N') \geq\alpha R_2 + (1-\alpha)R'_2 - \delta$. Therefore $(\alpha R_1 + (1-\alpha)R'_1, \alpha R_2 + (1-\alpha)R'_2 )$ is achievable and therefore $C_{\max}(\mathbf{W})$ is convex.
\end{proof}

\begin{lemma}[Non-empty maximal error capacity region]\label{max_error_non_empty}
	For CCQ channel $\mathbf{W}$, if $\text{int}(C(\mathbf{W}))$ is non-empty, then $\text{int}(C_{\max}(\mathbf{W}))$ is also non-empty.
\end{lemma}

\begin{proof}
	Assume $\text{int}(C(\mathbf{W}))\neq \emptyset$ and $(R_1, R_2) \in \text{int}(C(\mathbf{W}))$ and so $R_1\neq 0$ and $R_2 \neq 0$.  From an average error code that achieves the rate pairs $(R_1, R_2)$, we construct two codes that max-error achieve $(R_1, 0)$ and $(0, R_2)$ respectively. Since $C_{\max}(\mathbf{W})$ is convex and closed, by a time sharing argument, the interior will thus be non-empty. 
	
	By definition, since $(R_1, R_2)\in \text{int}(C(\mathbf{W}))$, there exists a $k_0$ such that for all $k\geq k_0$, the $(k,M,N)$ code $\co \coloneqq(x_m,y_n,D_{mn})_{m=1,n=1}^{M,N}$ satisfies, for all $\delta, \epsilon >0$, $\frac{1}{k}\log M \geq R_1 -\delta$, $\frac{1}{k}\log N \geq R_2 -\delta$, and $\overline{e}(\co,W^k)\leq \epsilon$. Using this code, we construct a code in the following way. Since $\overline{e}(\co, W^k) \leq \epsilon$, we can write,
	\begin{align*}
		\epsilon & \geq 1 - \frac{1}{MN}\sum_{\substack{m=1 \\n=1}}^{M,N} \tr(D_{mn} W^{k}(x_m, y_n)) \\ &= \frac{1}{N}\sum_{n=1}^N \left(1 - \frac{1}{M}\sum_{m=1}^M \tr(D_{mn} W^{k}(x_m, y_n)) \right).
	\end{align*}
	Therefore, there exists at least one index $n_0 \in [N]$ such that 
	\begin{align*}
		1 - \frac{1}{M}\sum_{m} \tr(D_{m{n_0}} W^{k}(x_m, y_{n_0})) \leq \epsilon. 
	\end{align*}
	With this $n_0$, we construct the code $(x_m, y_{n_0}, D_{m{n_0}})_{m=1}^M$ which achieves (under average error) $(R_1, 0)$. Note, the decoders no longer form a POVM, but can be modified to form a POVM via an expurgation as in \cite[Ch. 16.5]{wilde_text} with negligible effects for large $k$. We transform this code to a max-error achieve $(R_1, 0)$. Assume without loss of generality that the codewords $x_m$ are ordered such that $\left(1 - \tr(D_{m{n_0}} W^{k}(x_m, y_{n_0})) \right)_{m\in [M]}$ is non-decreasing. For $M' \coloneqq\ciel{2^{(k/2) (R_1-\delta)}}$, let $\lambda \coloneqq 1 - \tr(D_{{M'}{n_0}} W^{k}(x_{M'}, y_{n_0}))$. It holds,
	\begin{align*}
		\epsilon &\geq 1 - \frac{1}{M}\sum_{m=1}^M \tr(D_{m{n_0}} W^{k}(x_m, y_{n_0})) \\
		& \begin{aligned}
		\geq & \frac{1}{\ciel{2^{k(R_1 - \delta)}}}   \sum_{m=1}^{M'-1} \underbrace{1 - \tr(D_{m{n_0}} W^{k}(x_m, y_{n_0}))}_{\geq 0} \\  &+ \frac{1}{\ciel{2^{k(R_1 - \delta)}}}   \sum_{m=M'}^{\ciel{2^{k(R_1 - \delta)}}} \underbrace{1 - \tr(D_{m{n_0}} W^{k}(x_m, y_{n_0}))}_{\geq \lambda} \end{aligned}\\
		&\geq \frac{1}{\ciel{2^{k(R_1 - \delta)}}}(\ciel{2^{k/2(R_1 - \delta)}}) \lambda \\
		&= \lambda/2.
	\end{align*}
	Thus $\lambda \leq 2\epsilon$, $M' \rightarrow M$ as $k\rightarrow\infty$, and $\co'\coloneqq(x_m, y_{n_0}, D_{m{n_0}})_{m=1}^{M'}$ achieves maximal error rate $(R_1,0)$. We can make analogous steps to find another code $\co''$ that achieves maximal error rate $(0,R_2)$. By the convexity of $C_{\max}(\mathbf{W})$, we have that, for all $\alpha \in [0,1], (\alpha R_1, (1-\alpha) R_2)$ is an achievable maximal error rate and thus the interior of $C_{\max}(\mathbf{W})$ is not empty.
\end{proof}

\section{Message Identification}

In this section, we model message identification over a CCQ multiple access channel.

\begin{definition}[Randomized $(k,M)$-ID-code]
	For a CQ channel $\mathbf{W}$, a randomized $(k,M)$-ID-code is a family $\mathcal{C}_{\text{id}} \coloneqq (P_m, I_m)_{m=1}^{M}$, where $P_1, ..., P_M \in \mathcal{P}(\mathcal{X}^k)$ are probability distributions, and for each $m\in[M]$, $I_m \in \mathcal{L}(\mathcal{H}^{\otimes k})$ with $0\leq I_m \leq \mathbbm{1}_{\mathcal{H}^{\otimes k}}$. \\\\ 	
	For an $(k, M)$-ID-code $\cocid$ and CQ channel $\mathbf{W}$, we define two types of errors, 
	\begin{align*}
		e_1(\cocid, W^{k}) & \coloneqq \max_{m\in[M]}  1 - \sum_{x^k\in \mathcal{X}^k} P_m(x^k) \tr\left(I_m W^{ k}(x^k) \right), \\
		e_2(\cocid, W^{k}) & \coloneqq \max_{\substack{m, n \in [M]                                                               \\ m\neq n}} \sum_{x^k\in \mathcal{X}^k}  P_{m}(x^k)\tr(I_n W^{k}(x^k)).
	\end{align*}	
\end{definition}

\begin{definition}[$(k,M,N)$-ID-code]
	For a a CCQ channel $\mathbf{W}$, a $(k,M,N)$-ID-code for classical message identification is the family $\mathcal{C}_{\text{id}} \coloneqq (P_m, Q_n,I_{mn})_{m=1,n=1}^{M,N}$ where $P_1,...,P_M \in \mathcal{P}(\alx^k)$, $Q_1,...,Q_N\in\mathcal{P}(\aly^k)$, and $(I_{mn})_{m=1,n=1}^{M,N}\subseteq \mathcal{L}(\hi^{\otimes k})$ such that $0 \leq I_{mn} \leq \mathbbm{1}_{\hi^{\otimes k}}$ for all $m\in[M]$ and $n \in [N]$.
	\\\\\noindent
	For a $(k,M,N)$-ID-code $\cid$, we define two types of errors,
	\vspace{10mm}
	\begin{align*}
		  & e_1(\cid, W^k) \coloneqq \\ &\max_{\substack{m\in[M] \\ n\in[N]}}  1 -  \sum_{\substack{x^k\in\alx^k \\ y^k\in\aly^k}} P_m(x^k)Q_n(y^k)\tr\left( I_{mn}W^k(x^k,y^k) \right), \\
		&e_2(\cid, W^k) \coloneqq   \\
		&\max_{\substack{m,m'\in[M], \\ n,n'\in[N] \\                                                                                          \\ 
		(m,n)\neq (m',n')}}   \sum_{\substack{x^k\in\alx^k \\ y^k\in\aly^k}}P_m(x^k)Q_n(y^k)\tr\left(I_{m'n'}W^k(x^k,y^k) \right).
	\end{align*}		
\end{definition}

\begin{definition}[Simultaneous $(k,M,N)$-ID-code]
	A $(k,M,N)$-ID-code $\cocid \coloneqq (P_m, Q_n, I_{mn})_{m=1,n=1}^{M,N}$ is called simultaneous if for $R,S\in\mathbb{N}$ there exists a POVM $(E_{rs})_{r=1,s=1}^{R,S}$ with subsets $A_1, ..., A_M \subset [R]$ and $B_1, ..., B_N \subset [S]$ such that for each $m \in [M]$ and $n\in[N]$, 
	\begin{align*} 
		I_{mn} = \sum_{i\in A_m} \sum_{j \in B_n} E_{ij}. 
	\end{align*}
\end{definition}

\begin{definition}[Achievable ID-rate pair]
	For a CCQ channel $\mathbf{W}$, we say $(R_1, R_2) \in \mathbb{R}^2$, $R_1, R_2 \geq 0$, is an achievable ID-rate pair if for all $\epsilon_1, \epsilon_2, \delta > 0$, there exists a $k_0$ such that for all $k \geq k_0$, there is a $(k,M,N)$-ID-code $\cocid$ with 
	\begin{align*}
		\begin{aligned}
		\frac{1}{k} \log\log M \geq R_1 - \delta, & \hspace{3mm} \frac{1}{k}\log\log N \geq R_2 - \delta, \\  e_1(\cocid, W^k) \leq \epsilon_1,  &\hspace{3mm} e_2(\cocid, W^k) \leq \epsilon_2. 
		\end{aligned}
	\end{align*}
	The ID capacity region of $\mathbf{W}$ is defined as 		
	\begin{align*}
		\cid & (\mathbf{W}) \coloneqq \\ &\{(R_1, R_2) \mid (R_1, R_2) \text{ is an achievable ID-rate pair} \}. 
	\end{align*}
\end{definition}

\begin{definition}[Achievable simultaneous ID-rate pair] 
	For a CCQ channel $\mathbf{W}$, we say $(R_1, R_2) \in \mathbb{R}^2$, $R_1, R_2 \geq 0$, we say $(R_1, R_2)$ is an achievable simultaneous ID-rate pair if for $\epsilon_1, \epsilon_2, \delta > 0$, there exists a $k_0$ such that for all $k \geq k_0$ there is a simultaneous $(k,M,N)$-ID-code $\cosimcid$ with
	\begin{align*}
		\begin{aligned}
		\frac{1}{k} \log\log M \geq R_1 - \delta, & \hspace{3mm} \frac{1}{k}\log\log N \geq R_2 - \delta, \\  e_1(\cosimcid, W^k) \leq \epsilon_1,  &\hspace{3mm} e_2(\cosimcid, W^k) \leq \epsilon_2. 
		\end{aligned}
	\end{align*}
	The simultaneous ID capacity region for a CCQ channel $\mathbf{W}$ is defined as		
	\begin{align*}
		  & \cid^{\text{sim}}(\mathbf{W}) \coloneqq \\ &\{(R_1, R_2) \mid (R_1, R_2) \text{ is an achievable sim. ID-rate pair} \}. 
	\end{align*}
\end{definition}

\begin{remark}\label{rem_id}
	Since the simultaneous case is more restrictive, it is clear that 
	\begin{align*}
		\cid^{\text{sim}}(\mathbf{W}) \subseteq \cid(\mathbf{W}). 
	\end{align*}
\end{remark}

\begin{theorem}\label{main_result}
	For a DM-CCQ channel generated by $W:\alx\cross\aly\rightarrow\st(\hi)$, 
	\begin{align*}
		C(W) = \cid^{\text{sim}}(W) 
	\end{align*}
\end{theorem}

\section{Proof of Achievability}

Here we give the proof of achievable for Theorem (\ref{main_result}), and in the next section we give the converse proof. The proof here follows the approach taken by Ahlswede and Dueck in \cite{transformator}.

\begin{theorem}\label{sim_ach}
	For a DM-CCQ channel generated by $W: \alx\cross\aly\rightarrow\st(\hi)$, 
	\begin{align*}
		C(W) \subseteq \simcid(W). 
	\end{align*} 
\end{theorem}

\begin{lemma}[Chernov-H\"offding Bound]\label{corrrr}
	Let $M\in\mathbbm{N}$, and let a sequence of random variables $(\psi_i)_{i\in[M]}$ be such that for each $i$,  $\psi_i \in \{0,1\}$. Assume for each $i$ and $\mu,\lambda\in(0,1)$ the expectation value $\mathbb{E}(\psi_i) \leq \mu  < \lambda$. Then it holds,
	\begin{align*}
		\Pr\left( \sum_{j=1}^{M} \psi_j > \lambda M \right) \leq 2^{-M \cdot D(\lambda \parallel \mu)}. 
	\end{align*}
	where  $D(\lambda \parallel \mu)$ is the relative entropy between the probability distributions $(\lambda, 1 -\lambda)$ and $(\mu, 1 -\mu)$.
\end{lemma}

\begin{lemma}[Transformator lemma for a DM-CCQ channel]\label{transformator}
    For the DM-CCQ channel generated by $W:\alx\times\aly\rightarrow\st(\hi)$ with an achievable rate $(R_1, R_2) \in C(W)$, there is a simultaneous $(n, M', N')$-ID-code $\cocid^{\text{sim}}$ that achieves the simultaneous ID-rate $(R_1, R_2)$. 
\end{lemma}

\begin{proof}
	Assume for DM-CCQ generated by $W$ that $\text{int}(C(W))\neq \emptyset$, otherwise the problem is reduced to the single sender case, a problem solved in \cite{loeber_id_cap}. For a rate pair $(R_1, R_2) \in C(W)$, by definition, there exists a $k_0$ such that for all $k\geq k_0$, there is a $(k,M',N')$-code $\co' \coloneqq (u_i', v_j', D'_{ij})_{i = 1, j =1}^{M', N'}$ that achieves $(R_1,R_2)$	with $\overline{e}(\co', W^k) \leq \lambda(k)$, where $\lambda(k)\rightarrow 0$ as $k\rightarrow\infty$. Further, since int$(C(W))$ is non-empty, by Lemma \ref{max_error_non_empty}, there is a non-trivial achievable maximal error rate pair $(\epsilon_1,\epsilon_2) \in C_{\text{max}}(W)$ and thus a $k_0'$ such that for $k'\geq k_0'$, there is a $(\ciel{\sqrt{k}}, M'', N'')$-code $\co'' \coloneqq (u''_i, v''_j, D''_{ij})_{i=1,j=1}^{M'',N''}$, $M'' = 2^{\ciel{\sqrt{k}}\epsilon_1}$ and $N'' = 2^{\ciel{\sqrt{k}}\epsilon_2}$, with $e(\co'', W^{\ciel{\sqrt{k}}}) \leq \lambda(\sqrt{k})$, with $\lambda(\sqrt{k}) \rightarrow 0$ as $k\rightarrow\infty$, where the assumption that $\ciel{\sqrt{k}} \geq k_0'$ is made. We define $m\coloneqq k + \ciel{\sqrt{k}}$ and two families of maps 
	\begin{align*}
		\mathcal{A} & \coloneqq ( A_i : [M'] \rightarrow [M''])_{i=1}^{M}  \\
		\mathcal{B} & \coloneqq ( B_j : [N'] \rightarrow [N''])_{j=1}^{N}. 
	\end{align*}
	With these families of maps, we define an $(m, M, N)$-ID-code $(P_i, Q_j, I_{ij})_{i=1,j=1}^{M,N}$ where
	\begin{align*}
		P_i(x^m) \coloneqq 
		\begin{cases}
		\frac{1}{M'} & \text {if } \exists a \in [M'] : x^m = u_a'\cdot u''_{A_i(a)} \\
		0            & \text{otherwise}                                              
		\end{cases},
	\end{align*}
	and
	\begin{align*}
		Q_j(y^m) \coloneqq 
		\begin{cases}
		\frac{1}{N'} & \text {if } \exists b \in [N'] : y^m = v_b'\cdot v''_{B_j(b)} \\
		0            & \text{otherwise}                                              
		\end{cases}.
	\end{align*}
	The identifiers are defined as
	\begin{align*}
		I_{ij} \coloneqq  \sum_{a=1}^{M'}\sum_{b=1}^{ N'} D'_{ab} \otimes D''_{A_i(a)B_j(b)}. 
	\end{align*}
	We show that with this structure, there exists a random construction of $\mathcal{A}$ and $\mathcal{B}$ such that there is a simultaneous ID-code which achieves the simultaneous ID-rate pair $(R_1,R_2)$. 
	
	For $a\in[M']$ and $b\in[N']$ define the random variables $U_a$ such that 
	\begin{align*}
		\Pr(U_a = u_a'\cdot u_c'') = \frac{1}{M''} 
	\end{align*}
	with $c\in[M'']$, and $V_b$ such that \begin{align*}
	\Pr(V_b = v_b'\cdot v_d'') = \frac{1}{N''}
	\end{align*} 
	with $d \in [N'']$. For $(i, j) \in [M]\times [N]$, define
	\begin{align*}
		\ou{U}_i & \coloneqq \{ U_a \}_{a=1}^{M'}  \\
		\ou{V}_j & \coloneqq \{ V_b \}_{b=1}^{N'}, 
	\end{align*}
	Let $\overline{P}_i$ and $\overline{Q}_j$ be the uniform distributions on $\ou{U}_i$  and $\ou{V}_j$ respectively. Define the random identifier
	\begin{align*}
		\mathcal{I}(\ou{U}_i, \ou{V}_j) \coloneqq \sum_{a=1}^{M'}\sum_{b=1}^{N'} D(U_a, V_b) 
	\end{align*}
	where $D(U_a, V_b) \coloneqq D'_{ab} \otimes D''_{p_a q_b}$ when $U_a = u'_a \cdot u''_{p_a}$ and $V_b = v'_b\cdot v''_{q_b}$. With this, we can construct a random ID-code $(\overline{P}_i, \overline{Q}_j, \mathcal{I}(\ou{U}_i, \ou{V}_j)_{i = 1, j= 1}^{M,N}$. 
	
	We analyze the errors of realizations of the random code. Let $r\coloneqq \ciel{\sqrt{k}}$ and fix any two realizations $\mathcal{U}_i$ of $\ou{U}_i$ and $\mathcal{V}_j$ of $\ou{V}_j$.
	\begin{align*}
		1 &- \sum_{\substack{x^m \in \alx^m \\ y^m \in\aly^m }} P_i(x^m)Q_j(y^m) \tr( \mathcal{I}(\mathcal{U}_i, \mathcal{V}_j) W^m(x^m, y^m) ) \\ 
		&= 1 - \frac{1}{M'N'} \sum_{\substack{u \in \mathcal{U}_i \\ v\in \mathcal{V}_j}} \tr( \mathcal{I}(\mathcal{U}_i, \mathcal{V}_j) W^m(u, v) ) \\
		  & \begin{aligned} = 1 -    &   \\ &  
		\begin{aligned} 
		\frac{1}{M'N'} \sum_{\substack{u'_i\cdot u''_{p_i} \in \mathcal{U}_i \\ v'_j \cdot v''_{q_j} \in \mathcal{V}_j}} &\text{tr} \left[ \sum_{\substack{i'=1 \\ j'=1}}^{M', N'} \left(D'_{i'j'}\otimes D''_{p_{i'} q_{j'}}\right) \right. \\  
		&\left. \hspace{5mm} \cdot W^m(u'_i\cdot u''_{p_i}, v'_j \cdot v''_{q_j}) \right] 
		\end{aligned}
		\end{aligned} \\
		  & \begin{aligned} \leq 1 - &   \\ &  
		\begin{aligned} 
		\frac{1}{M'N'} \sum_{\substack{u'_i\cdot u''_{p_i} \in \mathcal{U}_i \\ v'_j \cdot v''_{q_j} \in \mathcal{V}_j}} &\text{tr} \left[ \left(D'_{ij}\otimes D''_{p_{i} q_{j}}\right) \right. \\  
		&\left. \hspace{5mm} \cdot W^m(u'_i\cdot u''_{p_i}, v'_j \cdot v''_{q_j}) \right] 
		\end{aligned}
		\end{aligned}\\
		&\begin{aligned}
		= 1 -  &\underbrace{\frac{1}{M'N'} \sum_{\substack{u'_i\cdot u''_{p_i} \in \mathcal{U}_i \\ v'_j \cdot v''_{q_j} \in \mathcal{V}_j}}  \tr( D'_{ij} W^k(u'_i, v'_j))}_{> 1 - \lambda(k)}  \\  
		&\hspace{20mm}\cdot\underbrace{\tr ( D''_{p_i q_j} W^r(u''_{p_i},  v''_{q_j}) )}_{> 1 - \lambda(\sqrt{k})}
		\end{aligned} \\
		&< \lambda(k) + \lambda(\sqrt{k}).
	\end{align*}	
	
    For the second type of error, we can start by considering the error between realizations $\mathcal{U}_1$ of $\ou{U}_1$ and $\mathcal{V}_1$ of $\ou{V}_1$ and random sets $\ou{U}_2$ and $\ou{V}_2$. We define two random functions, with the $i$th element of a realization of $\ou{U}_2$ denoted $U_i^2$ and similarly the $j$th element of $\ou{V}_2$ as $V_j^2$,
    \begin{align*}
        \psi_i (\ou{U}_2) \coloneqq 
        \begin{cases}
            1, & \text{if } U^2_i \in \mathcal{U}_1 \\
            0, & \text{otherwise}
        \end{cases}, \\
        \phi_j (\ou{V}_2) \coloneqq 
        \begin{cases}
            1, & \text{if } V^2_j \in \mathcal{V}_1 \\
            0, & \text{otherwise}
        \end{cases}.
    \end{align*}	
    Note that because each $U_i^2$ is independent of the other elements of $\overline{\mathcal{U}}_2$ and each $V_j^2$ is independent of the other elements of $\overline{\mathcal{V}}_2$, for $i\neq j$,  $\psi_i(\ou{U}_2)$ is independent of $\psi_j(\ou{U}_2)$  and  $\phi_j(\ou{V}_2)$ is independent of $\phi_i(\ou{V}_2)$. Further, it is easy to see $\forall i\in[M'], j\in[N']$,
    \begin{align*}
        \mathbb{E}(\psi_i(\ou{U}_2)) = \frac{1}{M''} \\
        \mathbb{E}(\phi_j(\ou{V}_2)) = \frac{1}{N''}, 
    \end{align*} 
    since the $\psi_i(\ou{U}_2) = 1$ when the ending of $U_i^2$ is equal to the ending of the $i$th element of $\mathcal{U}_1$ which occurs with $1/M''$ chance, and the analogous for $\phi_j(\ou{V}_2)$. For $\lambda\in(0,1)$ and that $\epsilon_1 = \log(M'')/r$, and $\epsilon_2 = \log(N'')/r$,
    \begin{align*}
        D&\left (\lambda \parallel \frac{1}{M''}\right )  \\ 
        &= \lambda \log(\lambda 2^{r\epsilon_1}) + (1-\lambda)\log(\frac{1-\lambda}{1-2^{-r\epsilon_1}}) \\
        & \begin{aligned} = \lambda &\log(\lambda ) + \lambda\log (2^{r\epsilon_1}) + (1-\lambda)\log(1-\lambda)
        \\ &- (1-\lambda)\log({1-2^{-r\epsilon_1}}) 	
        \end{aligned}\\
        &\geq \lambda\log (2^{r\epsilon_1}) + \log(0.5) \\
        &\geq \lambda  \sqrt{k} \epsilon_1 -1, 
    \end{align*}
    where we use the fact that $\lambda=0.5$ minimizes $\lambda \log(\lambda ) +(1-\lambda)\log(1-\lambda)$. Similarly, 
    \begin{align*}
        D\left (\lambda \parallel \frac{1}{N''}\right ) \geq \lambda \sqrt{k}\epsilon_2 - 1.
    \end{align*}
    Moreover, for any two realizations $\mathcal{U}_2$ and $\mathcal{V}_2$ when 
    \begin{align}\label{cond_1}
        \begin{aligned}
        (u_i \cdot u_{p_i} \notin \mathcal{U}_2 \text{ or } v_j\cdot v_{q_j} \notin \mathcal{V}_2)  \\ \nd (u_i \cdot u_{p_i} \in\mathcal{U}_1 \nd v_j\cdot v_{q_j} \in \mathcal{V}_1)
        \end{aligned}
    \end{align} 
    is true, it holds that,
    \begin{align*}
        \mathcal{I}(\mathcal{U}_2, \mathcal{V}_2) &= \sum_{\substack{i'=1\\j'=1}}^{M',N'} D'_{i'j'} \otimes D''_{p_{i'}q_{j'}} \\
        &\leq \sum_{\substack{i'=1 \\ j'=1}}^{M',N'} D'_{i'j'} \otimes \left( \mathbbm{1}_{\hi^{\otimes r}} - D''_{p_i q_j} \right) \\
        &\begin{aligned}
        =&D'_{ij} \otimes \left( \mathbbm{1}_{\hi^{\otimes r}} - D''_{p_i q_j}\right)  \\ 
        &+  \sum_{\substack{i'=1, i'\neq i \\ j'=1, j'\neq j}}D'_{i'j'} \otimes \left(\mathbbm{1}_{\hi^{\otimes r}} - D''_{p_i q_j}\right)
        \end{aligned} \\
        &\begin{aligned}
        =& D'_{ij} \otimes \left( \mathbbm{1}_{\hi^{\otimes r}} - D''_{p_i q_j} \right)  \\ 
        &+  \left(\mathbbm{1}_{\hi^{\otimes k}}-D'_{ij}\right) \otimes \left(\mathbbm{1}_{\hi^{\otimes r}} - D''_{p_i q_j}\right),
    \end{aligned}
    \end{align*}
        where the first inequality is true because when condition ($\ref{cond_1}$) holds, $D''_{p_{i'}q_{j'}}$ will never equal $D''_{p_{i}q_{j}}$. Since the $D''_{ij}$ form a POVM, the inequality holds. Therefore,
        \begin{align*}
        &\frac{1}{M'N'} \sum_{\substack{u \in \mathcal{U}_1 \setminus \mathcal{U}_2 \\ v \in \mathcal{V}_1 \setminus \mathcal{V}_2 } } \tr( \mathcal{I}(\mathcal{U}_2, \mathcal{V}_2)  W^{m}(u,v)) \\ 
        &\leq \frac{1}{M'N'} \sum_{\substack{u \in \mathcal{U}_1  \\ v \in \mathcal{V}_1 } } \tr( \mathcal{I}(\mathcal{U}_2, \mathcal{V}_2) W^{m}(u,v))\\ 
        & \begin{aligned}
        \leq\frac{1}{M'N'}\sum_{\substack{u'_i \cdot u''_{p_i} \in \mathcal{U}_1  \\ v'_j\cdot v''_{q_j} \in \mathcal{V}_1 }}&  \text{tr} \left[ \left(D'_{ij} \otimes \left(\mathbbm{1}_{\hi^{\otimes r}}- D''_{p_i q_j}\right)\right) \right.	\\ 
        &\hspace{2mm} \cdot W^m(u'_i \cdot u''_{p_i}, v'_j\cdot v''_{q_j})  \\
        & \hspace{-5mm} + \left( \left(\mathbbm{1}_{\hi^{\otimes k}} - D'_{ij}\right) \otimes 	\left(\mathbbm{1}_{\hi^{\otimes r}}- D''_{p_i q_j}\right)\right) \\
        & \left. \hspace{2mm}\cdot W^m(u'_i\cdot u''_{p_i},v'_j\cdot v''_{q_j}) \right]
        \end{aligned}\\	
        &  \begin{aligned}=
        &\frac{1}{M'N'} \sum_{\substack{u'_i \cdot u''_{p_i} \in \mathcal{U}_1  \\ v'_j\cdot v''_{q_j} \in \mathcal{V}_1 } } \tr(D'_{ij} W^k(u'_i,v'_j))   \\
        & \hspace{20mm} \cdot\tr((\mathbbm{1}_{\hi^{\otimes r}} - D''_{p_i q_j})W^r(u''_{p_i}, v''_{q_j}))  \\ 
        & +  \frac{1}{M'N'} \sum_{\substack{u'_i \cdot u''_{p_i} \in \mathcal{U}_1  \\ v'_j\cdot v''_{q_j} \in \mathcal{V}_1 } } \tr((\mathbbm{1}_{\hi^{\otimes k}} - D'_{ij}) W^k(u'_i,v'_j ))  \\ 
        &\hspace{20mm} \cdot \tr(( \mathbbm{1}_{\hi^{\otimes r}}- D''_{p_i q_j}) W^r(u''_{p_i}, v''_{q_j}))  \\
        \end{aligned}\\	
        &\leq \lambda(\sqrt{k}) + \lambda(k)\lambda(\sqrt{k})\\
        &=: \lambda_k, 
    \end{align*}
    With this, we have,
    \begin{align*}
        & \hspace{-5mm}\sum_{\substack{x^m \in \alx^m \\ y^m \in\aly^m}} P_1(x^m) Q_1(Y^m) \tr( \mathcal{I}(\ou{U}_2, \ou{V}_2) W^m(x^m, y^m )) \\ 
        &=  \frac{1}{M'N'} \sum_{\substack{u \in \mathcal{U}_1 \\ v \in \mathcal{V}_1 }} \tr(\mathcal{I}(\ou{U}_2, \ou{V}_2) W^m(u, v )) \\
        & \begin{aligned} &= \frac{1}{M'N'} \sum_{\substack{u \in \mathcal{U}_1 \cap \ou{U}_2  \\ v \in \mathcal{V}_1 \cap \ou{V}_2 }} \tr(\mathcal{I}(\ou{U}_2, \ou{V}_2) W^m(u, v ))
        \\ & \hspace{5mm}+ \frac{1}{M'N'} \sum_{\substack{u \in \mathcal{U}_1 \setminus \ou{U}_2  \\ v \in \mathcal{V}_1 \setminus \ou{V}_2 }}\tr(\mathcal{I}(\ou{U}_2, \ou{V}_2) W^m(u, v)) \\
        &  \hspace{5mm} + \frac{1}{M'N'} \sum_{\substack{u \in \mathcal{U}_1 \setminus \ou{U}_2  \\ v \in \mathcal{V}_1 \cap \ou{V}_2 }}\tr(\mathcal{I}(\ou{U}_2, \ou{V}_2) W^m(u, v)) \\
        &\hspace{5mm} +  \frac{1}{M'N'} \sum_{\substack{u \in \mathcal{U}_1 \cap \ou{U}_2  \\ v \in \mathcal{V}_1 \setminus \ou{V}_2 }}\tr(\mathcal{I}(\ou{U}_2, \ou{V}_2) W^m(u, v))
        \end{aligned}\\
        &\leq \frac{1}{M'N'} (|\mathcal{U}_1 \cap \ou{U}_2| \cdot |\mathcal{V}_1 \cap \ou{V}_2|)+ 3\lambda_k\\
        &=  \frac{1}{M'N'}\left( \sum_{i = 1}^{M'} \psi_i(\ou{U}_2) \cdot \sum_{j = 1}^{N'} \phi_j(\ou{V}_2) \right) + 3\lambda_k. 
    \end{align*}
    Using Lemma \ref{corrrr} twice, we have that for $\lambda \in (0,1)$ and $k$ large enough such that both $1/M'' < \lambda$ and $1/N'' < \lambda$, with non-zero probability, it holds that 
    \begin{align*}
        \frac{1}{M'}\sum_{i=1}^{M'} \psi_i(\ou{U}_2) < \lambda \hspace{4mm} \nd \hspace{4mm} \frac{1}{N'}\sum_{j=1}^{N'} \phi_j(\ou{V}_2) < \lambda.
    \end{align*}
    Therefore with non-zero probability, 
    \begin{align*}
        \frac{1}{M'N'}\left( \sum_{i = 1}^{M'} \psi_i(\ou{U}_2) \cdot \sum_{j = 1}^{N'} \phi_j(\ou{V}_2) \right) + 3\lambda_k \leq \lambda^2 + 3\lambda_k,
    \end{align*}
    which implies 
    \begin{align}\label{b0}
        \begin{aligned}
        \hspace{-5mm}\sum_{\substack{x^m \in \alx^m \\ y^m \in\aly^m }}P_1(x^m) Q_1(y^m) &\tr( \mathcal{I}(\ou{U}_2, \ou{V}_2) W^m(x^m, y^m )) \\& \hspace{20mm} \leq \lambda^2 + 3\lambda_k.
        \end{aligned}
    \end{align}
    Similar arguments can be made to show that,
    \begin{align}
        &\begin{aligned}
        \sum_{\substack{x^m \in \alx^m \\ y^m \in\aly^m }} \overline{P}_2(x^m) \overline{Q}_2(y^m) &\tr( \mathcal{I}(\mathcal{U}_1, \mathcal{V}_1) W^m(x^m, y^m ))\\ &\hspace{20mm}\leq \lambda^2 + 3\lambda_k \label{b1} 
        \end{aligned}\\
        &\begin{aligned}
        \sum_{\substack{x^m \in \alx^m \\ y^m \in\aly^m }} \overline{P}_2(x^m) Q_1(y^m) &\tr( \mathcal{I}(\mathcal{U}_1, \ou{V}_2) W^m(x^m, y^m )) \\ &\hspace{20mm}\leq \lambda^2 + 3\lambda_k \label{b2}
        \end{aligned}\\
        &\begin{aligned}
        \sum_{\substack{x^m \in \alx^m \\ y^m \in\aly^m }} P_1(x^m) \overline{Q}_2(y^m) &\tr( \mathcal{I}(\ou{U}_2, \mathcal{V}_1) W^m(x^m, y^m )) \\ &\hspace{20mm}\leq \lambda^2 + 3\lambda_k.\label{b3}
        \end{aligned}
    \end{align}
    Hence there exists a realizations of $\ou{U}_2$, $\mathcal{U}_2$ and $\ou{V}_2$, $\mathcal{V}_2$ such that  (\ref{b0}), (\ref{b1}), (\ref{b2}), and (\ref{b3}) are satisfied. We follow the argumentation of \cite{ahlswede_dueck_89}. With $\mathcal{U}_1$, $\mathcal{U}_2$, $\mathcal{V}_1$ and $\mathcal{V}_2$, with positive probability, it holds that $|\mathcal{U}_1 \cap \mathcal{U}_2| \leq \lambda M'$ and  $|\mathcal{V}_1 \cap \mathcal{V}_2| \leq \lambda N'$.  we would like to add two elements $\mathcal{U}_3$ and $\mathcal{V}_3$ such that $|\mathcal{U}_1 \cap \mathcal{U}_3| \leq \lambda M'$, $|\mathcal{U}_2 \cap \mathcal{U}_3| \leq \lambda M'$, $|\mathcal{V}_1 \cap \mathcal{V}_3| \leq \lambda N'$, and $|\mathcal{V}_2 \cap \mathcal{V}_3| \leq \lambda N'$, which implies that the second kind errors will hold for all elements in the code. We bound probability that such a $\mathcal{U}_3$ and $\mathcal{V}_3$ do not simultaneously exist with,
    \begin{align*}
    	\begin{aligned}
    	  & 2\cdot\Pr(\sum_{k=1}^{M'} \psi_k(\ou{U}_i) > \lambda M' ) \\ &+ 2 \cdot\Pr(\sum_{l=1}^{N'} \phi_l(\ou{V}_j) > \lambda N') < 1,
    	\end{aligned}		
    \end{align*}
    where the analogous of $\psi_k$ and $\phi_k$ are defined. We repeat the argument for $i = 4,...,M$ and $j=4,...,N$ and it should hold for existence that, 
    \begin{align*}
    	\begin{aligned}
    	  & (M-1)\cdot\Pr(\sum_{k=1}^{M'} \psi_k(\ou{U}_i) > \lambda M' ) \\ &+ (N-1) \cdot\Pr(\sum_{l=1}^{N'} \phi_l(\ou{V}_j) > \lambda N') < 1,
    	\end{aligned}		
    \end{align*}
    We can therefore enforce that for all $i = 2,...,M$, 
    \begin{align*}
    	(M-1)\cdot\Pr(\sum_{k=1}^{M'} \psi_k(\ou{U}_i) > \lambda M' ) < \frac{1}{2}, 
    \end{align*}
    and  all $j=2,...,N$ that,
    \begin{align*}
    	(N-1) \cdot \Pr(\sum_{l=1}^{N'} \phi_l(\ou{V}_j) > \lambda N') < \frac{1}{2}. 
    \end{align*}
    Thus it should hold that, for all $i = 2,...,M$ and $j=2,...,N$,   
    \begin{align*}
    	\Pr(\sum_{k=1}^{M'} \psi_k(\ou{U}_i) > \lambda M' ) \leq 2^{-M'(\lambda\sqrt{k}\epsilon_1 - 1)} < \frac{1}{2 (M-1)}, 
    \end{align*}
    and 
    \begin{align*}
    	\Pr(\sum_{l=1}^{N'} \phi_l(\ou{V}_j) > \lambda N')\leq 2^{-N'(\lambda\sqrt{k}\epsilon_2 - 1)}  < \frac{1}{2(N-1)}. 
    \end{align*}	
    With the choices 
    \begin{align*}
    	M \leq 2^{(2^{k(R_1-\delta)}(\lambda\sqrt{k}\epsilon_1 - 1)  - 1)/2} 
    \end{align*} 
    and       
    \begin{align*}
    	N \leq 2^{(2^{k(R_2-\delta)}(\lambda\sqrt{k}\epsilon_2 - 1)  - 1)/2}, 
    \end{align*}    
    there exists realizations of $(\ou{U}_1, ...,\ou{U}_M)$, and $(\ou{V}_1,..., \ou{V}_N)$ such that with the family of maps $\mathcal{A}$ and $\mathcal{B}$ constructed from them, there is a simultaneous $(m,M,N)$-ID-code $\cosimcid$ that achieves the simultaneous ID-rate $(R_1,R_2)$ as $k\rightarrow\infty$. In the limit as $k\rightarrow \infty$,  $m = k$ and thus the $\cosimcid$ is a simultaneous $(k,M,N)$-ID-code.
\end{proof}

\begin{proof}[Proof of Theorem \ref{sim_ach}]
	The proof follows directly from the Transformator lemma.
\end{proof}

\begin{corollary}
	For a DM-CCQ channel generated by $W: \alx\cross\aly \rightarrow \mathcal{S}(\hi)$, 
	\begin{align*}
		C(W) \subseteq \cid(W). 
	\end{align*}
\end{corollary}

\begin{proof}
	For a DM-CCQ generated by $W$, by Theorem \ref{sim_ach}, $C(W) \subseteq \simcid(W)$  By Remark \ref{rem_id}, we have that $\simcid(W) \subseteq \cid(W)$. Combining these results, it holds that $C(W) \subseteq \cid(W)$.
\end{proof}

\section{Proof of Converse}

\begin{notation}\label{notational_thing}
	For probability distribution $P^k\in\mathcal{P}(\alx^k)$ and CQ channel $\mathbf{W}$, we write 
	\begin{align}
		P^kW^k\coloneqq\sum_{x^k \in \alx^k} P^k(x^k)W^k(x^k), 
	\end{align} 
	or with an additional distribution $Q^k\in\mathcal{P}(\aly^k)$, for a CCQ channel $\mathbf{W}$, 
	\begin{align}
		P^kQ^kW^k \coloneqq\sum_{\substack{x^k \in \alx^k \\ y^k \in \aly^k}} P^k(x^k)Q^k(y^k)W^k(x^k,y^k).
	\end{align}
	When dealing with classical-classical channels $\mathbf{W} \coloneqq \{W^k(y^k|x^k):x^k\in \alx^k, y^k\in\aly^k\}_{k\in\mathbb{N}}$, we write for $y^k \in \aly^k$ the output of the channel,
	\begin{align}
		P^kW^k(y^k)  \coloneqq \sum_{x^k \in \alx^k} P^k(x^k)W^k(y^k|x^k). 
	\end{align} 
\end{notation}

\begin{notation}
	In some cases, for notational simplicity, we refer to a random variable by its distribution. For example, for a random variable $A$ with distribution $p$, we may refer to $A$ by $p$.
\end{notation}

\begin{definition}\label{variational_dis}
	Let $\rho \in \mathcal{S}(\hi)$ and $D \coloneqq \{ D_i \}_{i\in [M]}$ a POVM for $M\in \mathbb{N}$. Then a probability distribution $\rho(D)$ is induced on $[M]$ such that $\rho(D)(i) = \tr(\rho D_i)$ for $i \in [M]$. For a second state $\sigma \in \mathcal{S}(\hi)$, we define, with $D$,
	\begin{align*}
		d_D(\rho, \sigma) \coloneqq d_1 (\rho(D), \sigma(D)), 
	\end{align*}
	with $d_1$ the total variational distance, that is, for a set $A$ and two distributions $p,q \in \mathcal{P}(A)$
	\begin{align*}
		d_1(p,q) \coloneqq \sum_{a\in A} |p(a) - q(a)| = 2 \sup_{A'\subseteq A}\{p(A') - q(A')\}. 
	\end{align*}
\end{definition}	

\begin{lemma}[Steinberg \cite{steinberg}, Lemma 6]\label{steinberg_saviour}
	With a classical-classical channel $\mathbf{W} \coloneqq \{W^k(y^k|x^k):x^k\in \alx^k, y^k\in\aly^k\}_{k\in\mathbb{N}}$, for a fixed $R>0$, $\rho >0$, $k_0\geq 1$, assume that for every $k>k_0$ there exists a collection of distributions $\mathbf{P} \coloneqq \{ P_i^k \}_{i=1}^{N} \subseteq \mathcal{P}(\alx^k)$ such that, 
	\begin{align}
		\frac{1}{k} \log \log N \geq R, 
	\end{align}
	and, 
	\begin{align}
		\min_{i\neq j} d_1(P_i^k W^k, P_j^k W^k) > 2(1-\rho). 
	\end{align}
	Then for every $\gamma < (\rho/4)\cdot \min(1, R)$, there exists a subset $\mathbf{\tilde{P}} \subseteq \mathbf{P}$ such that for every $k > k_0(m,R,\rho,\gamma)$ independent of $\mathbf{W}$ and $\mathbf{P}$ it holds,
	\begin{align}
		|\mathbf{\tilde{P}}| \geq \exp \exp (kR) - \exp\exp(k(R - \gamma)) 
	\end{align}
	and for every $\tilde{P}^k \in \mathbf{\tilde{P}}$,
	\begin{align}
		R(1 - 4\rho) \leq \frac{1}{k} I(\tilde{P}^k; \tilde{P}^kW^k). 
	\end{align}
\end{lemma}

\begin{definition}\label{povm_thing}
	For $D \in \mathcal{L}(\hi)$ such that $0\leq D \leq \mathbbm{1}$, we define POVM $P(D) \coloneqq \{D, \mathbbm{1}-D\}$.
\end{definition}

\begin{lemma}\label{d1_ineq}
	Let $\rho, \sigma \in \mathcal{S}(\hi)$ and $E = (E_i)_{i\in [M]}$ be a POVM on Hilbert space $\hi$ with $M\in \mathbb{N}$ events. Let $A_1, A_2 \subset [M]$ such that $A_1\cap A_2 = \emptyset$ and $A_1 \cup A_2 = [M]$. Then, for $D\coloneqq\sum_{i \in A_1} E_i$,
	\begin{align*}
		d_1(\rho(E), \sigma(E)) \geq d_1(\rho(P(D)), \sigma(P(D))), 
	\end{align*}
	with $P(D) \coloneqq (D, \mathbbm{1}-D)$.
\end{lemma}

\begin{proof}
	Let $\rho, \sigma, E, A_1$ and $A_2$ be defined as in the lemma statement. Then,
	\begin{align*}
		d_1&(\rho(E), \sigma(E))  \\ &= \sum_{m \in [M]}|\rho(E)(m) -\sigma(E)(m)|  \\ 
		&\begin{aligned}
		=\sum_{m \in A_1}&|\rho(E)(m) - \sigma(E)(m)|   \\ &+ \hspace{-2mm}\sum_{m \in A_2}|\rho(E)(m) - \sigma(E)(m)|
		\end{aligned}  \\
		&\begin{aligned} 
		\geq |\sum_{m \in A_1}& \rho(E)(m) - \sigma(E)(m)| \\ &+ |\sum_{m \in A_2}\rho(E)(m) - \sigma(E)(m)|
		\end{aligned}  \\
		& \begin{aligned}
		=|\sum_{m \in A_1}& \tr(\rho E_{m}) - \tr(\sigma E_{m}) | \\ &+ |\sum_{m \in A_2}\tr(\rho E_{m}) - \tr(\sigma E_{m}) |
		\end{aligned}   \\
		&\begin{aligned}
		= | \tr(\rho D)& - \tr(\sigma D) |  \\ &+  | \tr(\rho (\mathbbm{1}-D)) - \tr(\sigma(\mathbbm{1}-D)) |
		\end{aligned}\\
		  & \begin{aligned} = | & \rho(P(D))(1) -  \sigma(P(D))(1) | \\ &+  |  \rho(P(D))(2) -  \sigma(P(D))(2) | \end{aligned}\\
		&= d_1(\rho(P(D)), \sigma(P(D))) 
	\end{align*}
\end{proof}

\begin{definition} For parties $A, B$ and $C$ and a CCQ channel $\mathbf{W}$,
	\begin{align*}
		C'(\mathbf{W}) \coloneqq  \text{cl}\left( \liminf_{k\rightarrow\infty} C_k(\mathbf{W})\right) 
	\end{align*}
	with 
	\begin{align*}
		\begin{aligned}
		C_k(\mathbf{W}) &\coloneqq  \\
		\bigcup_{\substack{p_1 \in \mathcal{P}(\alx^k) \\ p_2 \in \mathcal{P}(\aly^k)}}&	\left\{ (R_1, R_2) \mid   R_1 \leq \frac{1}{k} I(A^k;C^k)_{\gamma^{k}_2(p_1,p_2)}, \right. \\ & \hspace{15mm} \left.
		R_2 \leq \frac{1}{k} I(B^k;C^k)_{\gamma^{k}_2(p_1,p_2)}  \right\},
		\end{aligned}
	\end{align*}	
\end{definition}

\begin{theorem}\label{converse_theorem}
	For CCQ channel $\mathbf{W}$, 
	\begin{align*}
		\simcid(\mathbf{W}) \subseteq C'(\mathbf{W}). 
	\end{align*}
\end{theorem}

\begin{proof}
	Given a CCQ channel $\mathbf{W} \coloneqq \{ W^k: \alx^k \times \aly^k \rightarrow \st(\hi_C^{\otimes k}) \}$, let $(R_1, R_2) \in \simcid(\mathbf{W})$, then for all $\lambda_1, \lambda_2, \delta > 0$ there exists a $k_0$ such that for all $k \geq k_0$, there is a simultaneous $(k,M,N)-$ID-code $\cosimcid$  with 
	\begin{align}\label{code_thing}
		\begin{aligned}
		\frac{1}{k} \log \log M \geq R_1 - \delta, & \hspace{2mm} \frac{1}{k} \log \log N \geq R_2 - \delta, \\ \hspace{2mm} e_1(\cosimcid, W^k) \leq \lambda_1,& \hspace{2mm} e_2(\cosimcid, W^k) \leq \lambda_2,
		\end{aligned}
	\end{align} 
	and specifically there are codes with  $\lambda_1+\lambda_2<1$. Without loss of generality, assume for $\mu> 0$, 
	\begin{align*}
		R_2 = R_1 - \mu. 
	\end{align*} 
	Let $\lambda_1 + \lambda_2 < 1$ and define $\epsilon \coloneqq 1 - \lambda_1 - \lambda_2$. For $k\geq k_0$, let $\cosimcid \coloneqq (P_i^k, Q_j^k, I^k_{ij}) _{i=1,j=1}^{M,N}$ be the simultaneous $(k, M ,N)$-ID-code, where with $\delta > 0$,
	\begin{align}\label{restrions}
		M \coloneqq  \ciel{2^{2^{(k(R_1 -\delta))}}} \hspace{2mm} \text{ and }\hspace{2mm} N \coloneqq \ciel{ 2^{2^{(k(R_2 -\delta))}}}, 
	\end{align} 
	with errors as in (\ref{code_thing}). Define POVM $E^k$ indexed by $z^k\in\mathcal{Z}^k$ as the the common refinement of $\{ I^k_{ij}  \}_{i=1,j=1}^{M,N}$. For a fixed $j \in [N]$ define the channel, 
	\begin{align*}
		W^k_j: \alx^k \rightarrow \st(\hi_C^{\otimes k}), \hspace{2mm} x^k \mapsto \sum_{y^k \in \aly^k} Q_j^k(y^k) W^k(x^k, y^k). 
	\end{align*}
	Then it is easy to see that with $W^k_j$, $( P^k_i, I^k_{ij} )_{i=1}^{M}$ is an simultaneous $(k,M)$-ID-code, since rate and error bounds hold. Further, for all $1 \leq a < b \leq M$,
	\begin{align*}
		d_{E^k}(P_a^kW_j^k & , P_b^kW_j^k) \\ &= d_1(P_a^k W_j^k(E^k), P_b^kW_j^k(E^k)) \\
		&\geq d_1(P_a^kW_j^k(P(I^k_{aj})), P_b^kW_j^k(P(I^k_{aj})))\\
		&\geq 2\left(\tr(P_a^kW_j^kI^k_{aj}) - \tr(P_b^kW_j^kI^k_{aj})\right)\\
		&= 2\left(\tr(P_a^kQ_j^kW^kI_{aj}^k) - \tr(P_b^kQ_j^kW^kI_{aj}^k)\right)\\
		&> 2\left(1 - e_1(\cosimcid, W^k) - e_2(\cosimcid, W^k)\right) \\
		&\geq 2(1 - \lambda_1 - \lambda_2) \\
		&= 2\epsilon.
	\end{align*}
	The first inequality holds by Lemma \ref{d1_ineq}. The second inequality holds by using Definition \ref{variational_dis} with a fixed subset for the variational distance. The second equality comes from simply replacing $W_j^k$ with its definition. The third inequality holds by definition of the error types. The fourth inequality holds by the error bound restrictions imposed on $\cosimcid$. With this, we see if we construct the classical channel 
	\begin{align*}
		\tilde{W}_j^k(z^k|x^k) = \tr(W_j^k(x^k) E_{z^k}^k) 
	\end{align*}
	for some finite output alphabet $\mathcal{Z}^k$ and $z^k \in \mathcal{Z}^k$, then for $1 \leq a \leq M$,
	\begin{align*}
		P_a^k \tilde{W}_j^k(z^k) & = \sum_{x^k \in \alx^k} P_a^k(x^k) \tilde{W}_j^k(z^k|x^k)     \\
		                         & = \sum_{x^k \in \alx^k} P_a^k(x^k) \tr(W_j^k(x^k) E_{z^k}^k)  \\ 
		                         & =  \tr(\sum_{x^k \in \alx^k} P_a^k(x^k) W_j^k(x^k) E_{z^k}^k) \\ 
		                         & = P_a^k W_j^k (E^k)(z^k),                                     
	\end{align*}
	where the first equality is by notational choice in Notation \ref{notational_thing}, the second by replacing $\tilde{W}_j^k(z^k|x^k)$ with its definition, the third by the linearity of the trace and the last again by notational choice. Then it holds for all $1 \leq a < b \leq M$, 
	\begin{align*}
		d_1(P_a^k \tilde{W}_j^k, P_b \tilde{W}_j^k) =d_1(P_a^k W_j^k(E^k), P_b^kW_j^k(E^k)) > 2\epsilon, 
	\end{align*}
	where the inequality holds from above. Letting $\rho \coloneqq 1 - \epsilon$ and choosing $\gamma < \min\left (\mu, \rho/4, \rho/ 4 \cdot (R_1 - \delta)\right )$, by Lemma \ref{steinberg_saviour}, there exists a subset $\mathbf{\tilde{P}}_j \subset (P_i^k)_{i=1}^{M}$ such that 
	\begin{align*}
		|\mathbf{\tilde{P}}_j| \geq \exp \exp (k(R_1 - \delta)) - \exp\exp(k(R_1 - \delta - \gamma)) 
	\end{align*}
	and for all $\tilde{P}_{i^*}^k \in \mathbf{\tilde{P}}_j$,
	\begin{align*}
		(R_1 - \delta)(1 - 4\rho) \leq \frac{1}{k} I(\tilde{P}^k; \tilde{P}^k W_j^k(E^k)), 
	\end{align*}
	for all $k$ sufficiently large, depending only on $k_0, R_1 - \delta, \gamma,$ and $\rho$. We show that 
	\begin{align*}
		\bigcap_{j=1}^{N} \mathbf{\tilde{P}}_j \neq \emptyset, 
	\end{align*}
	which implies that regardless of $j$, we can always find such a $\tilde{P}_{i^*}^k$. In reference to the proof of Lemma 6 in \cite{steinberg}, a set $\mathcal{Z}^k_j$, for a fixed $j$, is defined as, 
	\begin{align*}
		\mathcal{Z}^k_j \coloneqq \left\{ P^k \in (P_i^k)_{i=1}^{M}  \mid \frac{1}{k}I(P^k; P^kW_j^k(E^k)) < R_1 -\delta \right\}, 
	\end{align*}
	and it is shown that $|\mathcal{Z}^k_j| \leq \exp \exp(k(R_1-\delta -\gamma))$ and $\mathbf{\tilde{P}}_j \coloneqq (P_i^k)_{i=1}^{M} \setminus \mathcal{Z}^k_j$. It clear then that
	\begin{align*}
		\left|\bigcup_{j=1}^{N} \mathcal{Z}^k_j\right| \leq \sum_{j=1}^{N} |\mathcal{Z}^k_j| \leq {N}\exp \exp(k(R_1-\delta -\gamma)). 
	\end{align*}
	Therefore, it holds that,
	\begin{align*}
		\left| \bigcap_{j=1}^{N} \mathbf{\tilde{P}}_j \right| &= M - \left|\bigcup_{j=1}^{N} \mathcal{Z}^k_j\right| \\
		&\geq M - {N}\exp \exp(k(R_1-\delta -\gamma)) \\
		  & \begin{aligned}\geq & \ciel{\exp\exp(k (R_1 - \delta))} - \\ & \hspace{2mm} \ciel{\exp\exp(k (R_2 - \delta))}\cdot \\ &\hspace{2mm}\exp \exp(k(R_1-\delta -\gamma)) 
		\end{aligned}\\
		& \begin{aligned}
		= &\ciel{\exp\exp(k (R_1 - \delta))} - \\ 
		&\hspace{2mm} \ciel{\exp\exp(k (R_1 - \delta - \mu))}\cdot \\ 
		&\hspace{2mm}\exp \exp(k(R_1-\delta -\gamma))
		\end{aligned} \\
		&> 0,
	\end{align*}
	for $k$ large enough, where the second inequality is by how $M$ and $N$ are defined, and the second equality is by the assumption that $R_2 = R_1 - \mu$. Therefore, there exists at least one distribution $\tilde{P}_{i^*}^k$ for all $j \in [N]$ with 
	\begin{align*}
		(R_1 - \delta)(1 - 4\rho) \leq \frac{1}{k} I(\tilde{P}_{i^*}^k; \tilde{P}_{i^*}^k W_j^k(E^k)). 
	\end{align*}
	With such a $\tilde{P}_{i^*}^k$, construct the channel,
	\begin{align*}
		V^k_{\tilde{P}}: \aly^k \rightarrow \st(\hi^{\otimes k}_C), \hspace{2mm} y^k \mapsto \sum_{x^k \in \alx^k} \tilde{P}_{i^*}^k(x^k) W^k(x^k, y^k). 
	\end{align*}
	Repeating the same argumentation, we can conclude that there is a subset $\mathbf{\tilde{Q}} \subset (Q_j^k)_{j=1}^{N}$ with 
	\begin{align*}
		|\mathbf{\tilde{Q}}| \geq \exp \exp (k(R_2 - \delta)) - \exp\exp(k(R_2 - \delta - \gamma)) > 0 
	\end{align*}
	and for any element $\tilde{Q}^k_{j^*} \in \mathbf{\tilde{Q}}$, where $j^*$ indicates the index of such a distribution, 
	\begin{align*}
		(R_2 - \delta)(1 - 4\rho) \leq \frac{1}{k} I(\tilde{Q}_{j^*}^k; \tilde{Q}_{j^*}^k V^k_{\tilde{P}}(E^k)). 
	\end{align*}
	Now, with channel state $ \alpha \in \st(\hi_A^{\otimes k}\otimes \hi_C^{\otimes k})$,
	\begin{align*}
		\alpha & \coloneqq \sum_{x^k \in \alx^k} \tilde{P}_{i^*}^k(x^k)\ketbra{x^k} \otimes   W_{j^*}^k(x^k) 
	\end{align*}
	it holds,
	\begin{align*}
		I(A^k, C^k)_{\alpha} & = \chi(\tilde{P}_{i^*}^k, W^k_{j^*})                                                             \\
		                     & \geq I_{\text{acc}}(\tilde{P}_{i^*}^k, W^k_{j^*})                                                \\
		                     & = \max_{\text{POVM } \tilde{E}^k} I(\tilde{P}_{i^*}^k,  \tilde{P}_{i^*}^kW^k_{j^*}(\tilde{E}^k)) \\	
		                     & \geq  I(\tilde{P}_{i^*}^k, \tilde{P}_{i^*}^kW_{j^*}^k(E^k))                                      
	\end{align*}	
	where $I_{\text{acc}}$ is the accessible information defined in \cite[(10.179)]{wilde_text}. The first equality holds since $\alpha$ is a channel state. The first inequality is the Holveo bound. The second equality is the definition of accessible information. Now, defining channel state $\beta \in \st(\hi_B^{\otimes k}\otimes \hi_C^{\otimes k})$
	\begin{align*}
		\beta & \coloneqq \sum_{y^k \in \aly^k} \tilde{Q}_{j^*}^k(y^k)\ketbra{y^k} \otimes V^k_{\tilde{P}}(y^k), 
	\end{align*}
	by the same reasoning, 
	\begin{align*}
		I(B^k, C^k)_\beta \geq I(\tilde{Q}_{j^*}^k,  \tilde{Q}_{j^*}^k V^k_{\tilde{P}}(E^k)). 
	\end{align*}	
	It is clear that that for the channel state $\gamma_2^{k} \in \st(\hi_A^{\otimes k} \otimes \hi_B^{\otimes k}\otimes \hi_C^{\otimes k})$, 
	\begin{align*}
		&\gamma_2^{k} \coloneqq \gamma_2^{k}(\tilde{P}_{i^*}^k,  \tilde{Q}_{j^*}^k) \\ 
		  & \begin{aligned} & =  \sum_{\substack{x^k \in \alx^k \\ y^k \in \aly^k}}\tilde{P}_{i^*}^k(x^k) \ketbra{x^k} \otimes \tilde{Q}_{j^*}^k(y^k) \ketbra{y^k} \\[1pt] & \hspace{4cm} \otimes W^k(x^k, y^k) 
		\end{aligned},
	\end{align*}
	it holds, 
	\begin{align*}
		I(A^k, C^k)_\alpha = I(A^k, C^k)_{\gamma_2^{k}}, 
	\end{align*}
	and
	\begin{align*}
		I(B^k, C^k)_\beta = I(B^k, C^k)_{\gamma_2^{k}}. 
	\end{align*}	
	Therefore, 
	\begin{align*}
		(R_1 - \delta)(1 - 4\rho) \leq \frac{1}{k} I(A^k; C^k)_{\gamma^k_2} 
	\end{align*}
	and
	\begin{align*}
		(R_2 - \delta)(1 - 4\rho) & \leq  \frac{1}{k} I(B^k, C^k)_{\gamma^k_2}. 
	\end{align*}
	Since $\delta$ and $\rho$ can be arbitrarily small, it implies that, 
	\begin{align*}
		(R_1, R_2) \in C'(\mathbf{W}). 
	\end{align*}
\end{proof}

\section{Equality of Capacity Regions}
In this section we show that the multi-letter capacity region $C'(W)$ for a DM-CCQ multiple access channel $W$ is indeed equal to the single letter  $C(W)$, concluding that simultaneous identification capacity for $W$ is equal to its transmission capacity. 

\begin{definition}
	For a CCQ channel $\mathbf{W}$ and parties $A,B$ and $C$, for $k\in\mathbb{N}$,
	\begin{align*}
		\begin{aligned}
		R_k(\mathbf{W})  \coloneqq \\ &  
		\begin{aligned}
		\hspace{-10mm}\bigcup_{\substack{p_1 \in \mathcal{P}(\alx^k) \\ p_2\in \mathcal{P}(\aly^k) }  }	\left\{ (R_1, R_2) \mid R_1 \leq \frac{1}{k} I(A^k ; C^k | B^k)_{\gamma^k_2(p_1,p_2)}, \right. \\ \left. R_2 \leq \frac{1}{k} I(B^k;C^k|A^k)_{\gamma^k_2(p_1,p_2)}, \right. \\ \left. \hspace{1mm} R_1+R_2 \leq \frac{1}{k}I(A^k,B^k;C^k)_{\gamma^k_2(p_1,p_2)} \right\}
		\end{aligned} 
		\end{aligned}
	\end{align*}
\end{definition}

\begin{lemma}\label{c_sub_r}		
	For a CCQ channel $\mathbf{W}$, for all $k \in \mathbb{N}$,
	\begin{align*}
		C_k(\mathbf{W}) \subseteq R_k(\mathbf{W}). 
	\end{align*}
\end{lemma}

\begin{proof}
	We show for any two of random variables $A$ and $B$ such that $A$ and $B$ are mutually independent, $C$ the channel output, and $\gamma_2$ any channel state, that it holds,
	\begin{align*}
		I(A;C)_{\gamma_2} & \leq I(A;C|B)_{\gamma_2}, \\ I(B;C)_{\gamma_2} &\leq I(B;C|A)_{\gamma_2},\\
		I(A,C)_{\gamma_2} + I(B,C)_{\gamma_2} &\leq  I(A,B;C)_{\gamma_2}.
	\end{align*}
	With $A$ and $B$ independent, $I(A;B)_{\gamma_2} = 0$. Moreover, $I(A;B|C)_{\gamma_2} \geq 0$. So 
	\begin{align*}
		I(A & ;C|B)_{\gamma_2} \\ &\geq I(A;C|B)_{\gamma_2} + I(A;B)_{\gamma_2} - I(A;B|C)_{\gamma_2} \\
		&\begin{aligned}
		&\hspace{-0.5mm} = H(A|B)_{\gamma_2} + H(C|B)_{\gamma_2} - H(A,C|B)_{\gamma_2} \\ 
		& \hspace{2mm} + H(A)_{\gamma_2} +H(B)_{\gamma_2}-H(A,B)_{\gamma_2}\\
		&\hspace{2mm} - H(A|C)_{\gamma_2}-H(B|C)_{\gamma_2}+H(A,B|C)_{\gamma_2}
		\end{aligned} \\
		&\begin{aligned}
		&= H(A,B)_{\gamma_2}-H(B)_{\gamma_2} + H(B,C)_{\gamma_2} \\
		& \hspace{2mm} - H(B)_{\gamma_2} - H(A,B,C)_{\gamma_2}+H(B)_{\gamma_2} \\ 
		& \hspace{2mm} - H(A)_{\gamma_2} +H(B)_{\gamma_2}-H(A,B)_{\gamma_2} \\ 
		& \hspace{2mm} - H(A,C)_{\gamma_2} + H(C)_{\gamma_2} \\ 
		& \hspace{2mm} - H(B,C)_{\gamma_2} + H(C)_{\gamma_2} + H(A,B,C)_{\gamma_2} \\
		& \hspace{2mm} - H(C)_{\gamma_2}
		\end{aligned} \\
		&= H(A)_{\gamma_2} + H(C)_{\gamma_2} - H(A,C)_{\gamma_2} \\
		&= I(A;C)_{\gamma_2},
	\end{align*}
	where all equalities hold simply by definition of entropy and mutual information. It can be  shown in a similar way that similarly $I(B;C|A)_{\gamma_2} = I(B;C)_{\gamma_2}$. Further,  
	\begin{align*}
		I & (A;C)_{\gamma_2} + I(B;C)_{\gamma_2}                       \\ 
		  & = H(A) - H(A|C)_{\gamma_2} + H(B) - H(B|C)_{\gamma_2}      \\
		  & = H(A,B)_{\gamma_2} - H(A|C)_{\gamma_2}- H(B|C)_{\gamma_2} \\
		  & \begin{aligned}                                            
		  & = H(A,B)_{\gamma_2} - H(A,B|C)_{\gamma_2}                  \\ 
		  & \hspace{5mm}+ H(B|A,C)_{\gamma_2} - H(B|C)_{\gamma_2}      
		\end{aligned} \\
		  & = I(A,B;C)_{\gamma_2} - I(A;B|C)_{\gamma_2}                \\
		  & \leq I(A,B;C)_{\gamma_2}.                                  
	\end{align*}
	Where the first equality is by definition of quantum mutual information and also the fact that for classical variables the Shannon entropy is equal to the von Neumann entropy. The second equality follows by the independence of $A$ and $B$ and again that the Shannon and von Neumann entropies are equal for classical variables. The third equality follows because, with some manipulation $ H(A|C)_{\gamma_2} = H(A,B|C)_{\gamma_2} - H(B|A,C)_{\gamma_2}$, seen as follows. By definition,
	\begin{align*}
		H(A,B|C)_{\gamma_2} & = H(A,B,C)_{\gamma_2} - H(C) 
	\end{align*}
	and,
	\begin{align*}
		H(B|A,C)_{\gamma_2} & = H(A,B,C)_{\gamma_2} - H(AC), 
	\end{align*}
	and so,
	\begin{align*}
		H(A&,B|C)_{\gamma_2} - H(B|A,C)_{\gamma_2}\\ 
		& 
		\begin{aligned}
		  & =H(A,B,C)_{\gamma_2} - H(C)_{\gamma_2} \\ &\hspace{3mm}- H(A,B,C)_{\gamma_2} + H(AC)_{\gamma_2} 
		\end{aligned} \\ 
		&= H(AC)_{\gamma_2}-H(C)_{\gamma_2} \\ 
		&= H(A|C)_{\gamma_2}.
	\end{align*}
	The fourth equality, holds because  $ I(A;B|C)_{\gamma_2} = H(B|C)_{\gamma_2}-H(B|A,C)_{\gamma_2}$. To see this we just need that $H(B|A,C)_{\gamma_2} = H(A,B|C)_{\gamma_2}-H(A|C)_{\gamma_2}$ and the rest  follows by definition. Explicitly,
	\begin{align*}
		H(A|C)_{\gamma_2}   & = H(AC)_{\gamma_2} - H(C)_{\gamma_2}     \\
		H(A,B|C)_{\gamma_2} & = H(A,B,C)_{\gamma_2} - H(C)_{\gamma_2}, 
	\end{align*}
	so,
	\begin{align*}
		H(A & ,B|C)_{\gamma_2} -H(A|C)_{\gamma_2}                                          \\ 
		    & = H(A,B,C)_{\gamma_2} - H(C)_{\gamma_2} - H(AC)_{\gamma_2} + H(C)_{\gamma_2} \\ 
		    & = H(A,B,C)_{\gamma_2} - H(AC)_{\gamma_2}                                     \\
		    & = H(B|A,C)_{\gamma_2}.                                                       
	\end{align*}
	The inequality follows from positivity of mutual information. Since all random variables are chosen arbitrarily and this holds for any channel state, the statement holds for all $k$ as was to show.
\end{proof}

\begin{lemma}\label{transmis_sub_r}
	For a CCQ channel $\mathbf{W}$,
	\begin{align*}
		C(\mathbf{W}) \subseteq \text{cl}\left(\liminf_{k\rightarrow\infty} R_k(\mathbf{W}) \right) 
	\end{align*}
\end{lemma}

\begin{proof}	
	The proof follows the structure of \cite[Theorem 1]{verdu}. Let $(R_1, R_2) \in C(\mathbf{W})$, and let $\epsilon\in(0,1)$, $\delta > 0$. By definition, there exists a $k_0$ such that for all $k\geq k_0$, there is a $(k,M,N)$-code $\co \coloneqq (x_m, y_n, D_{mn})_{m=1,n=1}^{M,N}$ such that, 
	\begin{align*}
		\frac{1}{k} \log M \geq R_1 - \delta, \hspace{2mm} \frac{1}{k} \log N \geq R_2 - \delta, \hspace{2mm} \overline{e}(\co, W^k) \leq \epsilon. 
	\end{align*} 
	Let $\co$ be such a code with $k$ fixed. Let $A^k$ and $B^k$ be independent random variables uniformly distributed on the codewords represented by $[M]$ and $[N]$ respectively. Let $C^k$ be the output of $W^k$ when $A^k$ and $B^k$ are sent, and let $(\hat{A},\hat{B})$ be random variables for the decoding of $C^k$. Then, the Markov chain $(A^k, B^k) \rightarrow C^k \rightarrow (\hat{A}, \hat{B})$ is formed. By Fano's inequality, it holds,
	\begin{align}
		H(A^k ,B^k | C^k)_{\gamma^k_2} & \leq 1 + \epsilon\log(MN)             \\
		H(A^k | C^k)_{\gamma_2^k}      & \leq 1 + \epsilon\log(M) \label{in1}  \\
		H(B^k | C^k)_{\gamma_2^k}      & \leq 1 + \epsilon\log(N) \label{in2}. 
	\end{align}
	where $\gamma_2^k$ is the channel state constructed  with $A^k$ and $B^k$. We use Lemma \ref{single_sender} for existence of codes for the single sender inequalities (\ref{in1}) and (\ref{in2}) with no loss of rate or accuracy. Now, since $A^k$ and $B^k$ are uniformly distributed,
	\begin{align*}
		I(A^k, B^k;C^k)_{\gamma_2^k} & \geq (1-\epsilon)\log (MN) - 1 \\
		I(A^k;C^k)_{\gamma_2^k}      & \geq (1-\epsilon)\log (M) - 1  \\
		I(B^k;C^k)_{\gamma_2^k}      & \geq (1-\epsilon)\log (N) - 1. 
	\end{align*}
	Since $A^k$ and $B^k$ are independent, it holds,
	\begin{align*}
		I(A^k; C^k | B^k)_{\gamma_2^{k}} & = I(A^k; C^k | B^k)_{\gamma_2^{k}} + \underbrace{I(A^k; B^k)_{\gamma_2^{k}}}_{=0} \\
		                                 & = I(A^k;C^k, B^k)_{\gamma_2^{k}}                                                  \\
		                                 & \geq I(A^k;C^k)_{\gamma_2^k}                                                      
	\end{align*}
	The second equality is by the chain rule property of quantum mutual information (see Appendix B). The inequality holds since for any quantum state $\rho$, 
	\begin{align*}
		I(A^k;C^k, B^k)_{\rho} & = H(A^k)_{\rho} - H(A^k|C^k,B^k)_\rho \\
		                       & \geq H(A^k)_{\rho} - H(A^k|C^k)_\rho  \\
		                       & =	I(A^k;C^k)_\rho,                    
	\end{align*}
	where we use that conditioning does not increase quantum entropy. Similarly it can be shown that, 
	\begin{align*}
		I(B^k; C^k | A^k)_{\gamma_2^{k}} & \geq I(B^k;C^k)_{\gamma_2^k} 
	\end{align*}
	Combining these results, it holds,
	\begin{align*}
		(1-\epsilon)(R_1 -\delta) - \frac{1}{k}         & \leq \frac{1}{k} I(A^k; C^k | B^k)_{\gamma_2^{k}} \\
		(1-\epsilon)(R_2 -\delta) - \frac{1}{k}         & \leq \frac{1}{k} I(B^k; C^k | A^k)_{\gamma_2^{k}} \\
		(1-\epsilon)(R_1 + R_2 - 2\delta) - \frac{1}{k} & \leq \frac{1}{k} I(A^k, B^k; C^k)_{\gamma_2^{k}}  
	\end{align*}
	and therefore, 	
	\begin{align*}
		(1-\epsilon)(R_1 - \delta, R_2 - \delta) - \left( \frac{1}{k}, \frac{1}{k} \right) \in R_k. 
	\end{align*}
	for $k$ large enough which further implies,
	\begin{align}
		(1-\epsilon)(R_1 - 2\delta, R_2 - 2\delta)  \in \liminf_{k\rightarrow \infty} R_k.\label{lim_inf_ini} 
	\end{align}
	Since $\epsilon$ and $\delta$ are chosen arbitrarily, (\ref{lim_inf_ini}) gives us that $(R_1, R_2)$ is a limit of a sequence of points in $\liminf_{k\rightarrow \infty} R_k$ and is therefore in the closure of  $\liminf_{k\rightarrow \infty} R_k$ with channel state constructed with the distributions on $A^k$ and $B^k$.
\end{proof}

\begin{lemma}[Subadditivity of mutual information]\cite[Lemma 1]{winter_q_mac}\label{subadd}
	For a quantum channel $W: \st(\hi_A) \rightarrow \st(\hi_C)$, for all $k\in \mathbb{N}, k > 0$, and any channel state $\gamma^k \in \st(\hi_A^{\otimes k} \otimes \hi_C^{\otimes k})$, it holds
	\begin{align}
		I(A^k; C^k)_{\gamma^k} \leq k\cdot I(A;C)_{\gamma}. 
	\end{align}
	where $\gamma$ is $\gamma^k$ restricted to one instance of $A$ and $C$. 
\end{lemma}

\begin{lemma}\label{rk_sub_c}
	For a DM-CCQ channel generated by $W$, for all $k_0 \in \mathbb{N}$ it holds, 
	\begin{align*}
		\text{cl}\left( \bigcup_{k\geq k_0} R_k(W) \right) \subseteq C(W). 
	\end{align*}
\end{lemma}

\begin{proof}
	Let $k_0\in\mathbb{N}$ and $k \geq k_0$. Further, let $(R_1, R_2) \in R_k$. With this, it holds,  
	\begin{align*}
		R_1 & \leq \frac{1}{k}I(A^k;C^k|B^k)_{\gamma_2^k}  \\
		    & = \frac{1}{k} I(A^k;C^k,B^k)_{\gamma_2^k}    \\
		    & \leq \frac{1}{k}  \cdot kI(A;C,B)_{\gamma_2} \\
		    & = I(A;C|B)_{\gamma_2},                       
	\end{align*}
	with $\gamma_2^k$ the channel state constructed from $A^k$ and $B^k$. The first inequality is by definition of $R_k(W)$. For mutually independent random variables $A^k$ and $B^k$, it is easy with the chain rule for mutual information that, for $C^k$ the channel output space, $I(A^k;C^k|B^k)_\rho = I(A^k;B^k,C^k)_\rho$, for any state $\rho$ and so the first equality holds. The second inequality is due to Lemma \ref{subadd}. It can be similarly shown that $R_2 \leq I(B;C|A)_{\gamma_2}$. Further, Lemma \ref{subadd} gives $R_1+R_2 \leq I(A,B;C)_{\gamma_2}$. Thus, with the channel state $\gamma_2$, $(R_1,R_2)\in C(W)$. 
\end{proof}

Using these results, it is now straight forward to prove the following theorem.
\begin{theorem}\label{c_eq_rk}		
	For a DM-CCQ channel generated by $W$, it holds, 
	\begin{align*}
		C(W) =  \text{cl}\left( \liminf_{k\rightarrow\infty} R_k(W) \right). 
	\end{align*}
\end{theorem}

\begin{proof}
	Combining the results of the previous lemmas, we have,
	\begin{align*}
		C(W) & \subseteq \text{cl}\left( \liminf_{k\rightarrow\infty} R_k(W) \right) \\ &
		\subseteq \text{cl}\left( \limsup_{k\rightarrow\infty} R_k(W) \right) \\
		& \subseteq \text{cl}\left( \bigcup_{k\geq k_0} R_k(W) \right)\\
		& \subseteq C(W),		
	\end{align*}
	where the first containment is from Lemma \ref{transmis_sub_r}, the second and third by the structure of the limits for sets with any $k_0 > 0$, and the last by Lemma \ref{rk_sub_c}.
\end{proof}

Using these results, we can prove the main theorem of the review.

\begin{proof}[Proof of Theorem \ref{main_result}]
	It holds,
	\begin{align*}
		C(W) & \subseteq \cid^{\text{sim}}(W)                                        \\
		     & \subseteq \text{cl}\left( \liminf_{k\rightarrow\infty} C_k(W) \right) \\
		     & \subseteq \text{cl}\left( \liminf_{k\rightarrow\infty} R_k(W) \right) \\
		     & = C(W),                                                               
	\end{align*}
	where the first containment is by Theorem (\ref{sim_ach}), the second from Theorem \ref{converse_theorem}, the third by Lemma \ref{c_sub_r} and the equality by Lemma \ref{c_eq_rk}.
\end{proof}

\end{document}